\documentclass[11pt,a4paper]{amsart}
\usepackage[T1]{fontenc}
\usepackage[utf8]{inputenc}
\usepackage{lmodern,microtype}
\usepackage[margin=22mm]{geometry}
\usepackage{setspace}
\usepackage{amsmath,amssymb,amsthm}
\usepackage{booktabs}
\usepackage[numbers,sort&compress]{natbib}
\usepackage{xurl}
\usepackage[unicode,hidelinks]{hyperref}
\hypersetup{pdftitle={Stability of the Courtade--Kumar inequality},
            pdfauthor={Vu Khac Ky and Tuan Tran}}
\newcommand{\E}{\mathbb E}
\newcommand{\PP}{\mathbb P}
\DeclareMathOperator{\Var}{Var}
\DeclareMathOperator{\atanh}{atanh}
\newtheorem{theorem}{Theorem}[section]
\newtheorem*{cktheorem}{Theorem}
\newtheorem{lemma}[theorem]{Lemma}
\newtheorem{proposition}[theorem]{Proposition}
\newtheorem{corollary}[theorem]{Corollary}
\numberwithin{equation}{section}
\allowdisplaybreaks[2]
\title[Stability of the Courtade--Kumar inequality]
{Stability of the Courtade--Kumar inequality}

\author{Vu Khac Ky}
\address{Department of Mathematics, FPT University, Hanoi, Vietnam}
\email{kyvk2@fpt.edu.vn}

\author{Tuan Tran}
\address{School of Mathematical Sciences, University of Science and Technology of China, Anhui, China}
\thanks{Tuan Tran was supported by the Excellent
Young Talents Program (Overseas) of the National Natural Science Foundation of China under Grant No.
GG0010007003.}
\email{trantuan@ustc.edu.cn}

\begin{document}
\raggedbottom

\begin{abstract}
We prove dimension-independent stability for the Courtade--Kumar
inequality: a Boolean function $f:\{-1,1\}^n\to\{-1,1\}$ whose information is close to the
dictator value is close in probability to a signed dictator. The correlation dependence is sharp in order near zero and,
for increasing functions, also at the noiseless endpoint.
\end{abstract}
\maketitle
\pagestyle{plain}

\section{Introduction}\label{sec:introduction}
How much information can a one-bit summary of independent fair bits
carry about a noisy observation of those bits? The Courtade--Kumar conjecture \cite{CourtadeKumar2014}, a longstanding problem in information theory, asserted that keeping one input bit is optimal. We study the
corresponding stability question: if a summary retains nearly as much
information, how often can it differ from one of the original bits
or its complement? Our bounds are independent of the number of input
bits.

To formulate the question, let $X$ be uniform on $\{-1,1\}^n$,
where $n\ge1$. For $0\le\rho\le1$, its noisy observation
$Y_\rho\in\{-1,1\}^n$ is obtained by flipping each coordinate of $X$
independently with probability $(1-\rho)/2$.
A one-bit summary is $f(X)$, where
$f:\{-1,1\}^n\to\{-1,1\}$ is Boolean. Keeping one coordinate,
possibly changing its sign, corresponds to a \emph{signed dictator}
$\sigma X_i$, with $1\le i\le n$ and $\sigma\in\{-1,1\}$. We call $f$ \emph{increasing} if $f(x)\le f(y)$ whenever
$x_i\le y_i$ for every $i$. It is \emph{balanced} if its two
values are equally likely under the uniform input.

We measure uncertainty by \emph{binary Shannon entropy},
$H_{\rm b}(p)=-p\log p-(1-p)\log(1-p)$ for a Bernoulli bit with
success probability $p\in[0,1]$. All logarithms are natural, and $0\log0=0$.
For a $\{-1,1\}$-valued bit with mean $t\in[-1,1]$, we write
$H(t)=H_{\rm b}((1-t)/2)$ for its entropy and
$\Phi(t)=\log2-H(t)$ for its entropy deficit relative to a fair bit.
The latter is also the \emph{relative entropy} of its law with
respect to the uniform law on $\{-1,1\}$.
Both $H$ and $\Phi$ take values in $[0,\log2]$.

The \emph{mutual information} $I(f(X);Y_\rho)$ measures the reduction
in this uncertainty: it is the entropy of $f(X)$ minus its conditional
entropy given $Y_\rho$. For a signed dictator these entropies are
$\log2$ and $H(\rho)$, respectively, so its information is $\Phi(\rho)$.

The optimality of this value was conjectured by Courtade and Kumar
\cite{CourtadeKumar2014} and proved independently in our companion paper,
\emph{Dictators are most informative} \cite{VuTranCK}, and in
\cite{IndependentCK2026}. We refer to our companion paper for
detailed background and prior work.

\begin{cktheorem}[{\cite{VuTranCK,IndependentCK2026}}]
For every Boolean function $f:\{-1,1\}^n\to\{-1,1\}$ and every
$0\le\rho\le1$, with $X,Y_\rho$ as above,
\[
 I(f(X);Y_\rho)\le I(X_1;Y_\rho)=\Phi(\rho).
\]
\end{cktheorem}

The following stability theorem, announced in our companion paper \cite{VuTranCK},
is the main result of this work. Its proof develops an entropy-barrier
approach, together with quantitative estimates that track the
information deficit and the effect of compression. We use analytic
tools from the companion paper, but do not assume the
Courtade--Kumar inequality.

\begin{theorem}[Stability]\label{st:natural-intro}
Let $0<\rho<1$ and $\varepsilon\ge0$.
If a Boolean function
$f:\{-1,1\}^n\to\{-1,1\}$ satisfies
$I(f(X);Y_\rho)\ge\Phi(\rho)-\varepsilon$, then there are
$1\le i\le n$ and $\sigma\in\{-1,1\}$ such that
\[
 \PP\left(f(X)\ne\sigma X_i\right)
 \le\frac{10^8\varepsilon}{\rho^2(1-\rho)H(\rho)}.
\]
If $f$ is increasing, one can choose $i$ so that
\[
\PP\left(f(X)\ne X_i\right)\le\frac{10^8\varepsilon}{\rho^2H(\rho)}.
\]
\end{theorem}

Taking $\varepsilon=0$ shows that equality in the Courtade--Kumar
bound holds precisely for signed dictators when $0<\rho<1$.
The endpoint dependence is intrinsic: at $\rho=1$, every balanced
Boolean function carries $\log2$ information, including functions
far from dictators.

Three-variable majority shows that both bounds have optimal
order as $\rho\downarrow0$, and that the increasing-function
bound also has optimal order as $\rho\uparrow1$
(Section~\ref{sec:global-stability}).

\subsection{Method and contributions}

Stability requires more than ruling out a violation of the
Courtade--Kumar inequality: the information deficit must control
the distance to a dictator. We use entropy production to follow
the decrease of conditional entropy as correlation increases.
The key additional requirement is a quantitative surplus over
the dictator's production rate even when the conditional entropy
exceeds $H(\rho)$, with a controlled loss depending on the
information deficit. This larger entropy changes the profile heights used in the production estimates, but a single dilation factor controls all the changes: one plus the
positive information deficit divided by $H(\rho)$.
Two energy comparisons hold under an arbitrary entropy budget:
the spectral estimate uses Fourier energy, while the
marked-coordinate estimate separates a dominant coordinate from the
remaining variance. Converting these comparisons to stability loses
only a term linear in the deficit.

We combine these estimates by comparing the deficit with a constant
multiple of $H(\rho)$ times the distance to a nearest signed dictator.
If it fell below this curve, the production
surplus would force the difference to keep decreasing as correlation
approaches one. This contradicts the nonnegative deficit
$\Phi(\E f)$ at the noiseless endpoint. The comparison in
Section~\ref{sec:principles} works pointwise in correlation, so a source
may move between the local and spectral regions. Quantitative
compression then transfers stability from increasing to arbitrary
functions by accounting for the information gained during sorting.

At lower correlations we use the local and Fourier estimates of
Samorodnitsky \cite{Samorodnitsky2016} and Yu \cite{Yu2023,Yu2026Local},
in the forms established in \cite{VuTranCK}, together with its strict
profile-clock comparisons. All inputs are stated explicitly. The
parameter checks are collected in the appendices, and the constants
are chosen for convenience. The verification repository is
\url{https://github.com/vukhacky/CK-stability};
Appendix~\ref{sec:certificates} describes the certificates and replay
procedure.

\section{Noise and entropy profiles}\label{sec:preliminaries}

This section collects the shared notation and analytic inputs from
\cite{VuTranCK}. Local and spectral estimates appear beside their
applications in Sections~\ref{sec:local-stability} and~\ref{sec:radius-tail}.

For a Boolean function $f:\{-1,1\}^n\to\{-1,1\}$, write
$A_\rho(f)=I(f(X);Y_\rho)$. Its information deficit and distance
to the nearest signed dictator are
\begin{equation}\label{eq:deficit-distance}
 \Delta_\rho(f)=\Phi(\rho)-A_\rho(f),\qquad
 \delta(f)=\min_{1\le i\le n,\,\sigma=\pm1}
              \PP\left(f(X)\ne\sigma X_i\right).
\end{equation}

\subsection{Fourier coefficients, noise, and entropy production}
Expectations and norms use the uniform probability measure on the
cube. Write $[n]=\{1,\ldots,n\}$ and, for $K\subseteq[n]$,
$X_K=(X_i)_{i\in K}$.
The products $\chi_S$ below form the Fourier--Walsh orthonormal
basis. A product involving $k$ coordinates has level $k$. In this
basis, the expansion of $f$ and the noise operator are
\[
 f=\sum_{S\subseteq[n]}\widehat f(S)\chi_S,
 \quad \chi_S(x)=\prod_{i\in S}x_i,
 \quad \widehat f(S)=\E[f\chi_S],
 \quad T_\rho f=\sum_S\rho^{|S|}\widehat f(S)\chi_S.
\]
The expansion and $T_\rho$ apply to all real-valued functions on
the cube. For Boolean $f$, the function $T_\rho f$ takes values in $[-1,1]$.
Noise multiplies every coefficient at level $k$ by $\rho^k$.
We write $\mu=\widehat f(\varnothing)=\E f$ for the mean and
$W_k(f)=\sum_{|S|=k}\widehat f(S)^2$ for the Fourier mass at level $k$.
When $f$ is fixed, we abbreviate $W_k(f)$ to $W_k$.
The largest singleton coefficient in absolute value is
$\alpha=\max_i|\widehat f(i)|$, where
$\widehat f(i)=\widehat f(\{i\})$. Parseval gives $\sum_kW_k=1$ and
$\Var(T_\rho f)\le\rho^2(1-\mu^2)$. The identity
$\E[f(X)\sigma X_i]=1-2\PP\left(f(X)\ne\sigma X_i\right)$ gives
$\delta(f)=(1-\alpha)/2$. These are the conventions of
\cite[Chapters~1--2]{ODonnell2014}.

For $g=T_\rho f$, the posterior mean of the bit is $g$, so its
conditional entropy is $E_\rho(f)=\E H(g)=H(\mu)-A_\rho(f)$.
Consequently,
\begin{equation}\label{ub:gap-definitions}
 \Delta_\rho(f)=E_\rho(f)-H(\rho)+\Phi(\mu).
\end{equation}
At $\rho=1$, the observation determines $f(X)$, so
\begin{equation}\label{eq:noiseless-deficit}
 E_1(f)=0,\qquad \Delta_1(f)=\Phi(\mu)\ge0.
\end{equation}
The term $\Phi(\mu)$ accounts for the smaller initial entropy of
an unbalanced bit. We will repeatedly use
\begin{equation}\label{ub:entropy-series}
 \Phi(t)=\sum_{k\ge1}\frac{t^{2k}}{2k(2k-1)},\qquad
 t^2/2\le\Phi(t)\le(\log2)t^2\quad(-1\le t\le1).
\end{equation}

For $h:\{-1,1\}^n\to\mathbb R$, write
$x_{-i}=(x_j)_{j\ne i}$ and let
$\partial_i h=[h(1,x_{-i})-h(-1,x_{-i})]/2$ and
$\mathcal L_i h=x_i\partial_i h$. The Dirichlet energy and its
coordinate parts are
\[
 \operatorname{Dir}(h)=\sum_i\E(\partial_i h)^2
 =\sum_S|S|\widehat h(S)^2,
 \qquad \operatorname{Dir}_i(h)=\E(\partial_i h)^2.
\]
For $g:\{-1,1\}^n\to(-1,1)$, define
$\mathcal D_i(g)=\E[(\mathcal L_i g)\atanh g]$ and
$D(g)=\sum_i\mathcal D_i(g)$. These are the coordinate and total
entropy-production energies. Fourier differentiation and $H'=-\atanh$
give
\begin{equation}\label{st:production-identity}
 -\rho E'_\rho(f)=D(T_\rho f),\qquad
 -\Delta'_\rho(f)=D(T_\rho f)/\rho-u,\qquad u=\atanh\rho.
\end{equation}
Thus $D(T_\rho f)/\rho-u$ is the production surplus over the dictator.
The parameter $u$ tends to infinity as correlation approaches the
noiseless endpoint. All derivatives are finite for a nonconstant
Boolean source and $0<\rho<1$.

For later endpoint estimates, the identity
$H(\tanh u)=u(1-\tanh u)+\log(1+e^{-2u})$ gives
\begin{equation}\label{eq:entropy-endpoint}
 u(1-\rho)<H(\rho)<(u+1)(1-\rho),\qquad
 H(\rho)\le(2u+1)e^{-2u},\qquad \rho=\tanh u.
\end{equation}

For $K\subseteq[n]$, put $P_Kh=\E[h\mid X_K]$ and
$\operatorname{Dir}_K(h)=\sum_{i\in K}\operatorname{Dir}_i(h)$.
Projection retains the Fourier terms supported inside $K$.
For Boolean $f$, the \emph{influence} of coordinate $i$ is
$\operatorname{Inf}_i(f)=\E|\partial_i f|$, the probability that
flipping coordinate $i$ changes $f$.
If $f$ is increasing, then
$\operatorname{Inf}_i(f)=\|\partial_i f\|_1=\widehat f(i)\ge0$.
In particular, $W_1(f)>0$ for every nonconstant increasing $f$.

\subsection{Entropy profiles and their clocks}

For a centered edge with values $-\tanh v,\tanh v$, the ratio of
entropy to half the edge difference is $F(v)$. We use three related
functions:
\begin{equation}\label{eq:profiles}
 F(v)=\frac{H(\tanh v)}{\tanh v},\qquad
 B_e(q)=qF^{-1}(e/q),\qquad
 \ell_u(a)=F^{-1}(F(u)/a).
\end{equation}
Here $u,v,e,q,a>0$, and we extend $B_e$ and $\ell_u$ to zero by
$B_e(0)=\ell_u(0)=0$. In applications $0\le a\le1$, so
$0\le\ell_u(a)\le u$.
The function $F$ is a decreasing bijection of $(0,\infty)$ onto
itself. In $B_e(q)$, $e$ is the entropy budget and $q$ is the
coordinate amplitude: the inverse gives the corresponding height,
and multiplication by $q$ gives the production bound.
The height $\ell_u(a)$ uses the dictator budget $H(\rho)$ and
amplitude $\rho a$, where $\rho=\tanh u$.

\begin{lemma}[Profile shape {\cite[Lemmas~4.1, 6.3, and~7.7]{VuTranCK}}]
\label{as:profile-convex}
The function $F$ is strictly decreasing and convex; $vF(v)$ is
strictly decreasing and log-concave. The function $1/[vF(v)^r]$
is convex for $0\le r\le1$. Also $-(\log F)'(v)>1$, and it
exceeds $7/4$ for $v\ge4$.
For the arcsine correction
\begin{equation}\label{eq:arcsine-correction}
 R(v)=v-\frac{\arcsin^2(\tanh v)}{\tanh v},\qquad R(0)=0,
\end{equation}
one has $0\le R'(v)\le1$ and $0\le R''(v)<1/2$.
\end{lemma}
By Lemma~\ref{as:profile-convex}, $1/[vF(v)]$ is increasing and
convex. Hence $F(u)/F(v)=F(u)v/[vF(v)]$ is increasing and convex
in $v$, and its inverse $a\mapsto\ell_u(a)$ is increasing and concave.
Also $\ell_u(a)/a=vF(v)/F(u)$ decreases, where $v=\ell_u(a)$.

The elementary entropy formula also gives
\begin{equation}\label{rt:scaled-profile}
 2v+1\le e^{2v}F(v)
 \le\frac{1+e^{-2v}}{1-e^{-2v}}(2v+1).
\end{equation}
The derivative bound for $R$ measures how much the arcsine refinement
can lose when the actual entropy is slightly larger than $H(\rho)$.

For a fixed profile $Q(y)=\sum_{j=1}^Jc_jB_y(q_j)$, with
$J\ge1$, $c_j>0$, and $0<q_j\le1$, define its clock and the source profile by
\begin{equation}\label{eq:profile-clock}
 T_Q(v)=\int_0^v\frac{dy}{y+Q(y)},\qquad
 P_f(y)=\sum_iB_y(|\widehat f(i)|).
\end{equation}
The weights and arguments in $Q$ stay fixed as correlation varies.
A dictator has profile $Q(y)=B_y(1)$; substitution of $y=H(r)/r$
gives $T_Q(H(\rho)/\rho)=\log(1/\rho)$. This identifies the
reference noise time in the comparison below.
\begin{theorem}[Clock propagation {\cite[Theorem~4.6]{VuTranCK}}]
\label{thm:clock-propagation}\label{st:clock}
Let $f:\{-1,1\}^n\to \{-1,1\}$. Suppose $W_1(f)>0$, $Q\le P_f$ on $(0,\infty)$, and
$0<\rho_0<1$. If a constant $\epsilon\in[0,1)$ satisfies
\[
 T_Q(H(\rho_0)/\rho_0)
 \le(1-\epsilon)\log(1/\rho_0),
\]
then, for $\rho_0\le\rho<1$,
$E_\rho(f)\ge H(\rho)\rho^{-\epsilon}$.
\end{theorem}

\subsection{Coordinate production and martingale energy}

The production inputs hold for the actual conditional entropy,
before comparison with a dictator. For nonconstant $f$ and
$0<\rho<1$, put $g=T_\rho f$ and $E=E_\rho(f)$. Then
\cite[Proposition~4.3 and the coordinate argument in its proof,
Proposition~6.2]{VuTranCK} give
\begin{equation}\label{st:actual-production}
 \mathcal D_i(g)\ge\rho|\widehat f(i)|
 F^{-1}\!\left(\frac{E}{\rho|\widehat f(i)|}\right),
\end{equation}
and, for every coordinate set $K$,
\begin{equation}\label{st:actual-arcsine}
 D(g)\ge\operatorname{Dir}(\arcsin g)
 +\rho\sum_{i\in K}|\widehat f(i)|
 R\!\left(F^{-1}\!\left(\frac{E}{\rho|\widehat f(i)|}\right)\right).
\end{equation}
Zero singleton coefficients contribute zero. We also have
$\mathcal D_i(g)\ge\operatorname{Dir}_i(g)$.

The marked-coordinate calculation needs an energy bound for a
conditional projection. The same parameterized inequality will apply
to the two conditional sections and to the remaining noisy coordinates.
This is the martingale entropy argument
of Falik and Samorodnitsky \cite{FalikSamorodnitsky2007}, with the
energy of the revealed coordinate retained.
\begin{lemma}[Martingale energy {\cite[Lemma~7.3]{VuTranCK}}]
\label{as:strengthened-split}
For $h:\{-1,1\}^n\to\mathbb R$ and $K\subseteq[n]$, let
$V=\|h-P_Kh\|_2^2$ and
$b_K=\sum_{i\notin K}(\E|\partial_i h|)^2$.
For every real $\lambda$,
\[
 \operatorname{Dir}(h-P_Kh)
 \ge(1+\lambda/2)V-e^{\lambda-1}b_K/2.
\]
\end{lemma}

\section{From entropy estimates to stability}\label{sec:principles}

Three principles organize the stability proof. A production criterion
compares the deficit with a multiple of $H(\rho)$.
Profile dilation keeps that deficit in every coordinate estimate.
Quantitative compression then controls the distance lost when
passing to an increasing function.

\subsection{Comparison with an entropy barrier}

Write $\rho=\tanh u$ and measure the deficit relative to the
dictator's remaining uncertainty:
\begin{equation}\label{eq:normalized-deficit}
 Z_f(u)=\frac{\Delta_{\tanh u}(f)}{H(\tanh u)}.
\end{equation}
Write $s_+=\max\{s,0\}$ for the positive part of a real number.
In the criterion below, $a$ measures the production surplus and $B$
its loss per unit of normalized deficit. The production estimate is
needed only below the threshold $Z_f(u)<\eta\delta(f)$; above it,
the desired lower bound already holds. The proof compares
$\Delta_\rho(f)$ with the barrier $cH(\rho)\delta(f)$ using the
endpoint identity \eqref{eq:noiseless-deficit}.

\begin{proposition}[From production to stability]\label{prop:production-stability}
Fix a nonconstant function $f:\{-1,1\}^n\to \{-1,1\}$, 
and let $u_0,a,\eta>0$ and $B\ge0$. Suppose that whenever
$u\ge u_0$ and $Z=Z_f(u)<\eta\delta(f)$, with $\rho=\tanh u$,
\begin{equation}\label{eq:production-criterion}
 (1+Z_+)\frac{D(T_\rho f)}\rho
 \ge u(1+a\delta(f)-BZ_+).
\end{equation}
Then $\Delta_\rho(f)\ge cH(\rho)\delta(f)$ for all
$\tanh u_0\le\rho<1$ and every
$0<c<\min\{\eta,1,a/(B+3)\}$.
\end{proposition}

\begin{proof}
Put $\delta=\delta(f)$. The case $\delta=0$ is immediate. For $\delta>0$, set
$G(\rho)=\Delta_\rho(f)-cH(\rho)\delta$.
This function is continuous on $[\tanh u_0,1]$.
By \eqref{eq:noiseless-deficit}, $G(1)=\Phi(\E f)\ge0$.
Suppose that $G$ is negative at some $\rho\ge\tanh u_0$.
Wherever $G<0$, we have $Z<c\delta<\eta\delta$ and
$Z_+<c\delta<1$, so the hypothesis applies. Since $s\mapsto(1+a\delta-Bs)/(1+s)$ is decreasing,
\eqref{eq:production-criterion} gives
\[
 \frac{D(T_\rho f)}{\rho u}
 \ge\frac{1+(a-Bc)\delta}{1+c\delta}>1+c\delta.
\]
The strict inequality follows from
$a>(B+3)c>(B+2)c+c^2\delta$.
Equation~\eqref{st:production-identity} therefore yields
$G'(\rho)=u(1+c\delta)-D(T_\rho f)/\rho<0$.
Starting from a negative value, $G$ keeps decreasing and cannot
return to zero. Continuity at $\rho=1$ then contradicts $G(1)\ge0$.
\end{proof}

Both production estimates below have this form: spectral transfer
supplies the surplus away from dictators, and a marked coordinate
supplies it near them. The criterion allows these estimates to
alternate with a direct local bound; no single region must persist
as correlation increases.

\subsection{A common dilation for all coordinates}

The ideal heights in \eqref{eq:profiles} use entropy $H(\rho)$.
The actual entropy can exceed $H(\rho)$ by at most the positive part
of the deficit. The following
comparison keeps that error in a single factor, which will be used in
both the marked-coordinate and spectral arguments.

\begin{lemma}\label{st:dilation}
For $c\ge1$ and $v>0$,
$F^{-1}(cF(v))\ge v/c$.
Let $f:\{-1,1\}^n\to \{-1,1\}$ be a nonconstant Boolean function and $0<\rho<1$.
Put $u=\atanh\rho$, $g=T_\rho f$, and
$\ell_i=\ell_u(|\widehat f(i)|)$.
If $\vartheta\ge1$ satisfies
$E_\rho(f)\le\vartheta H(\rho)$, then
\[
 \vartheta\mathcal D_i(g)\ge\rho|\widehat f(i)|\ell_i,\qquad
 \sum_i|\widehat f(i)|\ell_i\le\vartheta D(g)/\rho,
\]
and
\[
 \vartheta D(g)\ge\operatorname{Dir}(\arcsin g)
              +\rho\sum_{i\in K}|\widehat f(i)|R(\ell_i)
\]
for every coordinate set $K\subseteq[n]$.
In particular, one can take
$\vartheta=1+(\Delta_\rho(f))_+/H(\rho)$.
\end{lemma}

\begin{proof}
The decrease of $vF(v)$ gives $F(v/c)\ge cF(v)$, proving the
inverse bound. The entropy bound $E_\rho(f)\le\vartheta H(\rho)$
therefore makes each actual profile height at least $\ell_i/\vartheta$. Equation~\eqref{st:actual-production} gives
the coordinate and total production bounds. Moreover $0\le R'(v)\le1$, so
$R(\ell_i/\vartheta)\ge R(\ell_i)-(1-1/\vartheta)\ell_i$.
Apply \eqref{st:actual-arcsine} and pay the total error
using the total production bound. Its cost is at most
$(\vartheta-1)D(g)$, proving the arcsine assertion.
The stated choice of $\vartheta$ follows from
$E_\rho(f)=H(\rho)-\Phi(\mu)+\Delta_\rho(f)$.
\end{proof}

\subsection{Quantitative compression}

The classical compression reduction of Courtade and Kumar
\cite[Lemma~2]{CourtadeKumar2014} shows that sorting coordinate
fibers increasingly does not decrease mutual information.
For a Boolean function $f$, sorting once in each coordinate
produces an increasing Boolean function $g$ satisfying
\[
 \E g=\E f,\qquad
 A_\rho(g)\ge A_\rho(f),\qquad
 \delta(g)\le\delta(f)
 \qquad(0<\rho<1).
\]
See \cite[Lemma~2.1]{VuTranCK} for this formulation, including
the comparison of dictator distances. These comparisons reduce
information maximization to increasing functions, but do not
by themselves transfer stability: a bound on $\delta(g)$ need
not control $\delta(f)$. The next lemma supplies the missing
quantitative estimate, controlling the original dictator distance
by the compressed distance and the information gained during
sorting.

\begin{lemma}\label{st:compression}
Let $g$ be obtained from $f:\{-1,1\}^n\to\{-1,1\}$ by sorting each
coordinate increasingly, once per coordinate. For $0<\rho<1$,
\[
 \rho^2(1-\rho^2)\delta(f)
 \le A_\rho(g)-A_\rho(f)
       +\rho^2(2-\rho^2)\delta(g).
\]
Consequently, if $c>0$ and $\Delta_\rho(g)\ge c\delta(g)$ for every
increasing Boolean $g$, then every Boolean $f$ satisfies
\[
 \Delta_\rho(f)\ge
 \rho^2(1-\rho^2)
 \min\left\{1,\frac{c}{\rho^2(2-\rho^2)}\right\}\delta(f).
\]
\end{lemma}
\begin{proof}
Choose $i$ with $\widehat g(i)=1-2\delta(g)$. Immediately before
coordinate $i$ is sorted, write the current source as $b+x_i d$.
Sorting any other coordinate preserves the $i$th singleton coefficient.
Thus $\E d=\widehat f(i)$ and $\E|d|=\widehat g(i)$.
Also $d\in\{-1,0,1\}$, so $\E d^2=\widehat g(i)$.

At each noisy output of the other coordinates, the function
$v\mapsto[\Phi(T_\rho b+\rho v)+\Phi(T_\rho b-\rho v)]/2$
is even and has second derivative at least $\rho^2$ on its feasible
interval. Sorting replaces $T_\rho d$ by $T_\rho|d|\ge|T_\rho d|$.
The information gained in this step is therefore at least
\[
 \frac{\rho^2}{2}
 \left(\|T_\rho|d|\|_2^2-\|T_\rho d\|_2^2\right).
\]
Here the operators act on the other coordinates. Jensen's inequality
and the spectral gap bound the expression in parentheses from below by
\[
 (1-\rho^2)(1-\widehat f(i)^2)
 -(1-\widehat g(i))(1+\widehat g(i)-\rho^2).
\]
Now $1-\widehat f(i)^2\ge2\delta(f)$,
$1-\widehat g(i)=2\delta(g)$, and
$1+\widehat g(i)-\rho^2\le2-\rho^2$.
Substituting these bounds and adding the nonnegative gains from
all other sorting steps proves the first assertion.

For the second, write $L=A_\rho(g)-A_\rho(f)\ge0$.
Since $\Delta_\rho(f)=L+\Delta_\rho(g)$, the first assertion gives
\[
 \rho^2(1-\rho^2)\delta(f)
 \le L+\frac{\rho^2(2-\rho^2)}c\Delta_\rho(g)
 \le\max\{1,\rho^2(2-\rho^2)/c\}\Delta_\rho(f). \qedhere
\]
\end{proof}

\section{Local and intermediate correlations}\label{sec:local-stability}

We first describe the deficit when $f$ approaches a dictator at a
fixed correlation. A relative-entropy calculation gives the leading
term and an explicit remainder for arbitrary mean. In the balanced
case, the leading coefficient agrees with Yu's asymptotic bound
\cite[Theorem~4.2]{Yu2026Local}.

\begin{corollary}[Local stability]\label{cor:local-stability}
Let $f:\{-1,1\}^n\to\{-1,1\}$ and $0<\rho<1$, and put
$\delta=\delta(f)$ and $\mu=\E f$. Then
\[
 0\le A_\rho(f)-\Phi(\rho)+2\rho\atanh(\rho)\delta+\Phi(\mu)
 \le\frac{4\delta^{2/(1+\rho^2)}}{1-\rho^2}.
\]
\end{corollary}
\begin{proof}
Orient a closest dictator as $x_i$ and write $g=T_\rho f$.
For $s\in[-1,1]$ and $t\in(-1,1)$, the expression
$\Phi(s)-\Phi(t)-\atanh(t)(s-t)$ is the relative entropy
between Bernoulli laws of means $s$ and $t$. It lies between zero
and $(s-t)^2/(1-t^2)$: the upper bound follows from
$\log x\le x-1$. Average with $s=g(X)$ and $t=\rho X_i$.
Since $\E[X_i(g-\rho X_i)]=-2\rho\delta$, the middle expression
is exactly the resulting relative entropy. Hypercontractivity
\cite[Theorem~3.1]{VuTranCK} gives
$\|T_\rho(f-x_i)\|_2^2\le\|f-x_i\|_{1+\rho^2}^2
=4\delta^{2/(1+\rho^2)}$, proving the upper bound.
\end{proof}
Since $|\mu|\le2\delta$ and $\Phi(\mu)\le(\log2)\mu^2$, the
information deficit is
$2\rho\atanh(\rho)\delta+O_\rho(\delta^{2/(1+\rho^2)})$
as $\delta\downarrow0$, uniformly in the dimension. In particular,
the leading coefficient is sharp, and the deficit is at least
$\rho\atanh(\rho)\delta$ for sufficiently small $\delta$ at fixed $\rho$.
This local expansion assumes neither balance nor monotonicity.

\subsection{Local entropy and low-correlation inputs}

To obtain bounds on explicit dictator neighborhoods, we use a
comparison that retains one coordinate before applying entropy
contraction to the others. Define
\begin{equation}\label{eq:local-comparison}
 M_\rho(a)=(1-\rho^2)H(\rho a)-(1-\rho^2a)H(\rho),
 \qquad 0\le a\le1.
\end{equation}
It is a computable lower bound for the information deficit. Its
concavity turns an estimate at one coefficient into a linear bound
throughout a neighborhood of the dictators.

\begin{lemma}[Local comparison and propagation
{\cite[Theorem~3.4 and Lemma~3.5]{VuTranCK}}]
\label{ub:local}\label{ub:local-propagation}
For every Boolean function $f:\{-1,1\}^n\to \{-1,1\}$, every $i\in [n]$, and $0<\rho<1$,
$\Delta_\rho(f)\ge M_\rho(|\widehat f(i)|)$.
The function $M_\rho$ is concave and $M_\rho(1)=0$.
If $M_{\rho_0}(a)\ge0$ with $0<\rho_0<1$, then
\[
 M_\rho(a)/\rho^2\ge M_{\rho_0}(a)/\rho_0^2
 \qquad \text{for } 0<\rho\le\rho_0.
\]
\end{lemma}
In terms of distance, concavity and $M_\rho(1)=0$ give
\begin{equation}\label{eq:local-distance-conversion}
 \frac{\Delta_\rho(f)}{\delta(f)}
 \ge\frac{M_\rho(1-2d)}d
 \qquad(0<\delta(f)\le d\le1/2).
\end{equation}
Thus one positive value of $M_\rho$ gives a linear gap throughout
a dictator neighborhood. Outside that neighborhood, an absolute
gap suffices since $\delta(f)\le1/2$.

The balanced one-coordinate estimate is due to Yu
\cite[Theorem~4.6 and Remark~4.7]{Yu2026Local}; the stated form
retains the mean correction needed for arbitrary sources.
Two quantitative bounds in the proofs of the cited local estimates
will be useful near the noiseless endpoint:
\begin{lemma}[Local thresholds
{\cite[proofs of Lemmas~3.6 and~8.1]{VuTranCK}}]
\label{ub:exponential-local}\label{rt:local-junction}
If $u\ge\log9$, $\rho=\tanh u$, and $d=e^{-u}$, then
\[
 \frac{M_\rho(1-2d)}d>\frac{1-\rho}{120}.
\]
For $0<d\le1/64$, set $u_d=\log(3/(2d))$ and $\rho_d=\tanh u_d$.
Then
\[
 \frac{M_{\rho_d}(1-2d)}d>\frac{1-\rho_d}{80}.
\]
\end{lemma}
The first proof bounds $M_\rho(1-2d)/(d(1-\rho))$ below by
$71636/7971615>1/120$; the second bounds it below by
$65/4608>1/80$. Concavity extends each estimate to smaller distances,
and Lemma~\ref{ub:local-propagation} extends the second to smaller
correlations. We write the inverse threshold as
\begin{equation}\label{eq:local-thresholds}
 d_-(u)=\frac32e^{-u},\qquad
 d\le d_-(u)\ \Longleftrightarrow\ u\le u_d.
\end{equation}

We also use information contraction,
\cite[Corollary~3.3]{VuTranCK}:
\begin{equation}\label{ub:information-contraction}
 A_{\eta r}(f)\le\eta^2 A_r(f),\qquad
 A_\rho(f)\le\rho^2H(\mu)
 \quad(0\le\eta,r,\rho\le1).
\end{equation}
For $\alpha\le69/100$ and $|\mu|\le61/100$, the quadratic
Fourier estimate in \cite[proofs of Proposition~3.15 and
Lemma~A.2]{VuTranCK} gives
\begin{equation}\label{st:low-energy-input}
 \Delta_\rho(f)>\frac3{200000}
 \quad\left(\frac35\le\rho\le\frac{457}{500}\right),\qquad
 A_{3/5}(f)<\frac9{50}-\frac1{1000}.
\end{equation}
The scalar margins yielding these bounds are recorded in
Appendix~\ref{app:low-margins}. Thus contraction and Fourier estimates
supply an absolute gap away from dictators, while
\eqref{eq:local-distance-conversion} handles their neighborhood.

\begin{corollary}[Stability for $\rho\le457/500$]\label{cor:low-stability}
For every $f:\{-1,1\}^n\to\{-1,1\}$ and $0<\rho\le457/500$,
\[
 \Phi(\rho)-A_\rho(f)\ge\frac{\rho^2}{50000}\,\delta(f).
\]
\end{corollary}
\begin{proof}
The dictator case is immediate, so let $0<\delta=\delta(f)\le1/2$.
If $\alpha\ge69/100$, the local endpoint bound
$M_{457/500}(69/100)>(457/500)^2/1000$, recorded in
Appendix~\ref{app:low-margins}, propagates by
Lemma~\ref{ub:local-propagation} to
$M_\rho(69/100)>\rho^2/1000$.
Apply \eqref{eq:local-distance-conversion} with $d=31/200$.

If $|\mu|\ge61/100$, the entropy bound in
\eqref{ub:information-contraction} gives
$H(\mu)<1/2-1/1000$, so the deficit exceeds $\rho^2/1000$.
For the remaining sources with $\rho\ge3/5$, use
\eqref{st:low-energy-input}. Contraction of its anchor bound gives
a deficit at least $\rho^2/360$ for $\rho\le3/5$.
Each bound implies the assertion.
\end{proof}

\subsection{Intermediate correlations}\label{sec:cover}

Throughout this subsection, $f$ is increasing.
On the compact interval $457/500\le\rho\le49/50$, a positive
absolute gap away from dictators is enough. We extract this gap from
the source cover proved in \cite[Theorems~5.6 and~5.7]{VuTranCK}.
The five correlation bands and their rational parameters $M,A_*$
are recorded in Table~\ref{tab:middle-parameters}.
Each function is assigned to a comparison with a positive margin.
The local comparison applies when $\alpha\ge A_*$; the cubic-energy
comparison applies when
$\mathcal W_\rho=W_1+\rho^2W_2/(1+\rho)\le M$.
A third comparison uses a fixed profile $Q\le P_f$ with a sufficiently
short clock. In the last band, a two-coordinate moment comparison
also gives $\Delta_\rho(f)>(1-\rho^2)/1000$ on its assigned cells.
The source partitions and profile data are included with this paper.

The following strict margins are supplied by the same comparisons.
Writing $[\rho_-,\rho_+]$ for a band, the local and clock checks give
\begin{equation}\label{eq:middle-local-clock-slack}
 M_{\rho_+}(A_*)>\frac1{2000000},\qquad
 T_Q\bigl(H(\rho_-)/\rho_-\bigr)
 <\log(1/\rho_-)-\frac1{40000}.
\end{equation}
For the cubic bound, put
\begin{equation}\label{eq:middle-cubic-parameters}
 C_\rho=1+\rho+\rho^2-\rho^2(1+\rho)M,\qquad
 \kappa_\rho=\frac{H(\rho)}{(1-\rho)C_\rho}.
\end{equation}
Throughout the band, $C_\rho>0$,
$\kappa_\rho(1-\rho^3)<30/67$, and
\begin{equation}\label{eq:middle-entropy-slack}
 H(t)\ge\kappa_\rho(1-t)(1+t-\rho^3)+\frac{1-t}{8000}
 \qquad \text{for } 0\le t\le 1.
\end{equation}
The endpoint tests establish this strengthened entropy bound;
the propagation argument in \cite[Lemma~5.3]{VuTranCK} preserves
its added constant after division by $1-t$.
When the profile varies with the source coefficients, convexity
of its clock in the squared coefficients preserves the margin
between boundary checks. The script \path{check_stability_constants.py}
recomputes all these margins from the input data, as described in
Appendix~\ref{app:middle-margins}.

\begin{proposition}\label{st:middle}
For every increasing $f:\{-1,1\}^n\to \{-1,1\}$ and $457/500\le\rho\le49/50$, 
\[
\Delta_\rho(f)\ge\delta(f)/2000000.
\]
\end{proposition}
\begin{proof}
The dictator case is immediate, so assume $\delta(f)>0$ and put
$g=T_\rho f$.
On the cubic-energy branch, the Fourier estimate underlying the cover is
\begin{equation}\label{st:cubic-energy-input}
 \E[(1-|g|)(1+|g|-\rho^3)]
 \ge(1-\rho)C_\rho-(1-\rho^3)\mu^2.
\end{equation}
Average \eqref{eq:middle-entropy-slack} with $t=|g|$ and use
$\Phi(\mu)\ge\mu^2/2$. Since
$\E(1-|g|)\ge(1-\E g^2)/2\ge(1-\rho^2)(1-\mu^2)/2$, this gives
\[
 \Delta_\rho(f)\ge
 \frac{1-\rho^2}{16000}(1-\mu^2)+\frac7{134}\mu^2
 >\frac1{4000000}.
\]
Thus the scalar slack yields an absolute gap; the mean correction
also covers sources with small variance.

On the clock branch, $\log(1/\rho_-)<1/10$, so
\eqref{eq:middle-local-clock-slack} permits $\epsilon=1/4000$ in
Theorem~\ref{thm:clock-propagation}. It follows that
\[
 \Delta_\rho(f)\ge H(\rho)(\rho^{-1/4000}-1)
 >\frac1{4000000},
\]
because $H(\rho)>1/20$ and $\log(1/\rho)>1/50$.
The moment branch has a larger absolute gap.
Each of these bounds implies the required estimate since
$\delta(f)\le1/2$.

On the local branch, Lemma~\ref{ub:local-propagation} and
\eqref{eq:middle-local-clock-slack} give
$M_\rho(A_*)>(\rho/\rho_+)^2/2000000>1/4000000$.
The distance conversion in \eqref{eq:local-distance-conversion},
with $d=(1-A_*)/2\le1/2$, completes the proof.
\end{proof}

\section{Production bounds near the noiseless endpoint}
\label{sec:radius-tail}

We now establish the production estimates used in the entropy-barrier
argument. A spectral comparison handles functions far from dictators;
near a dictator, a marked-coordinate comparison preserves its profile
and uses the variance in the remaining coordinates. Both estimates
use the same entropy budget, and together they yield the stability
bound at the end of this section. Throughout this section, let
$u>0$ and write $\rho=\tanh u$ and $g=T_\rho f$.

\subsection{The spectral energy bound and its scalar cost}\label{sec:spectral-input}

The quadratic bound in \eqref{ub:entropy-series} leads to a
nonnegative remainder $\Psi$ and a useful variance identity:
\begin{equation}\label{eq:entropy-remainder}
 \begin{gathered}
 \Psi(t)=(\log2)t^2-\Phi(t),\qquad \mathcal J=\E\Psi(g),\\
 (\log2)\Var(g)=\Phi(\rho)-\Delta_\rho(f)-\Psi(\mu)+\mathcal J.
 \end{gathered}
\end{equation}
The entropy series gives $0\le\Psi(t)\le(\log2-1/2)t^2$.
Also $\Psi(t)=H(t)-(1-t^2)\log2\le H(t)$.
The marked-coordinate comparison will use $\mathcal J\ge0$ to bound
the variance from below. The spectral comparison retains this remainder:
the scalar estimate in \cite[Lemma~6.1]{VuTranCK} gives
$\|\arcsin g-(\pi/2)g\|_2^2\le C_*\mathcal J$, where
$C_*=(\pi/2-1)^2/(\log2-1/2)$.
Thus $\mathcal J$ controls the error in the linear approximation
to the arcsine transform. Define $\psi(0)=0$ and
\begin{equation}\label{eq:retained-variance}
 \psi(a)=\min\left\{\frac{a+a^2}{2},
       \frac{2a(1-a/4)}{\log(4/a)}\right\}\quad \text{for } 0<a\le 1.
\end{equation}
The retained-source estimate is
$\Var(P_Kf)\le\sum_{i\in K}\psi(\widehat f(i))$ for increasing $f$
\cite[Theorem~7.1]{VuTranCK}.

The parameter $t$ sets the Fourier weights, while $\gamma$ penalizes
the error in replacing $\arcsin g$ by $(\pi/2)g$. Choose $t\ge2$
and $\gamma>\max\{0,e^{2t-3}/2-t\}$, and define
\begin{equation}\label{eq:spectral-parameters}
 \begin{gathered}
 k(s)=\frac{\gamma s}{\gamma+s}\ (s>-\gamma),\qquad
 \Lambda=\frac{\pi^2}{4}k(t),\qquad \Gamma=\gamma C_*,\\
 \Delta_2=[k(t)-k(2)]\rho^4,\quad
 \Delta_1=[k(t)-k(1)]\rho^2,\quad
 \Delta_o=[k(t)-k(t-e^{2t-3}/2)]\rho^2.
 \end{gathered}
\end{equation}
The three costs $\Delta_2,\Delta_1,\Delta_o$ are nonnegative,
with $\Delta_1\ge\Delta_2$.
The next estimate separates the spectral baseline from the costs
of retaining or discarding coordinates. We include the Fourier
calculation from \cite[Section~7.3, Lemma~7.4 and
Proposition~7.5]{VuTranCK} to identify these costs before applying
the profile budget.
\begin{lemma}[Spectral energy]\label{st:spectral-energy}
Let $f:\{-1,1\}^n\to \{-1,1\}$ be an increasing nonconstant Boolean function, and let $K\subseteq[n]$. Then
\[
 \operatorname{Dir}(\arcsin g)\ge
 \Lambda\Var(g)-\Gamma\mathcal J-\frac{\pi^2}{4}\left\{
 \sum_{i\in K}[\Delta_2\psi(\widehat f(i))
       +(\Delta_1-\Delta_2)\widehat f(i)^2]
 +\Delta_o\sum_{i\notin K}\widehat f(i)^2\right\}.
\]
\end{lemma}
\begin{proof}
Put $h=\arcsin g$. Apply Lemma~\ref{as:strengthened-split} with
parameter $2t-2$ and add $\operatorname{Dir}(P_Kh)$.
Since $h$ is increasing, $\E|\partial_i h|=\widehat h(i)$.
The resulting Fourier weights are $\lambda_\varnothing=0$,
$\lambda_S=|S|$ for nonempty $S\subseteq K$,
$\lambda_{\{i\}}=t-e^{2t-3}/2$ for $i\notin K$, and $\lambda_S=t$
for all other modes. Each weight exceeds $-\gamma$.
Completing the square in each mode gives
\[
 \lambda_S x^2+\gamma(x-y)^2\ge k(\lambda_S)y^2.
\]
Use $x=\widehat h(S)$, $y=(\pi/2)\widehat g(S)$, and the scalar
error bound above to obtain
\[
 \operatorname{Dir}(h)\ge
 \frac{\pi^2}{4}\sum_S k(\lambda_S)\rho^{2|S|}\widehat f(S)^2
 -\gamma C_*\mathcal J.
\]
Relative to $k(t)\Var(g)$, retained and discarded singletons
cost $\Delta_1$ and $\Delta_o$, respectively. Every retained mode
of degree $j\ge2$ costs at most $\Delta_2$, because
$[k(t)-k(j)]\rho^{2j}\le\Delta_2$; all other modes have no cost.
The retained cost is therefore at most
$\Delta_2\Var(P_Kf)+(\Delta_1-\Delta_2)\sum_{i\in K}\widehat f(i)^2$.
Apply the retained-source variance bound to conclude.
\end{proof}

For a height threshold $0<z<u$, let $\beta=F(u)/F(z)$,
$K=\{i:\widehat f(i)\ge\beta\}$, and $\tau=\Delta_o\beta/z$.
Write $\ell_i=\ell_u(\widehat f(i))$. The decrease of
$\ell_u(a)/a$ gives $\widehat f(i)^2\le(\beta/z)\widehat f(i)\ell_i$
outside $K$, so the discarded mass is paid from the production budget.
For $0<a\le1$, the retained cost per unit of that budget is
\begin{equation}\label{eq:spectral-profile}
 \mathcal Q_u(a)=\frac{\rho[\ell_u(a)-R(\ell_u(a))]
       +(\pi^2/4)\Delta_2\psi(a)/a
       +(\pi^2/4)(\Delta_1-\Delta_2)a}{\ell_u(a)}.
\end{equation}
For an influence cutoff $0<\bar\alpha\le1$, define
\begin{equation}\label{st:spectral-cost}
 \begin{gathered}
 b=\max\left\{\rho+\frac{\pi^2}{4}\tau,
              \sup_{\beta\le a\le\bar\alpha}\mathcal Q_u(a)\right\},\\
 N=\Lambda\rho^2-\Gamma\Psi(\rho),\qquad m=N/b-u.
 \end{gathered}
\end{equation}
An empty supremum imposes no condition. The quantity $m$ measures
the production surplus above the dictator value $u$, while
$\Gamma/b$ will measure the loss per unit of deficit.

The scalar maximization in \eqref{st:spectral-cost} reduces to
height endpoints by the convexity estimates stated in
Appendix~\ref{app:scalar-cost}.

\subsection{Spectral comparison away from dictators}

Retain the parameters and scalar costs from
Section~\ref{sec:spectral-input}. The profile budget converts the
spectral energy estimate into a lower bound for production.
A transfer condition then compares that bound with the dictator
value while retaining the information deficit.

\begin{proposition}\label{st:robust-spectral}
Let $f:\{-1,1\}^n\to \{-1,1\}$ be an increasing nonconstant Boolean function, with
$0<\rho<1$ and $\max_i\widehat f(i)\le\bar\alpha$.
If $\vartheta\ge1$ and $E_\rho(f)\le\vartheta H(\rho)$,
then, for $g=T_\rho f$,
\[
 \vartheta\frac{D(g)}\rho
 \ge\frac{\Lambda\Var(g)-\Gamma\mathcal J}{b}.
\]
If the transfer condition
$\rho^2(\Gamma\log2-\Lambda)\ge\Gamma(\log2-1/2)$ holds,
this lower bound is at least $u+m-(\Gamma/b)\Delta_\rho(f)$.
In particular, $\vartheta=1+[\Delta_\rho(f)]_+/H(\rho)$ is admissible.
\end{proposition}
\begin{proof}
Put $P=\vartheta D(g)/\rho$ and
$w_i=\widehat f(i)\ell_u(\widehat f(i))$.
Lemma~\ref{st:dilation} gives $\sum_iw_i\le P$, and hence
\[
 \sum_{i\notin K}\widehat f(i)^2
 \le\frac\beta z\left(P-\sum_{i\in K}w_i\right).
\]
Combine the dilated arcsine estimate with
Lemma~\ref{st:spectral-energy} and use this bound on the discarded
coordinates to obtain
\[
 \Lambda\Var(g)-\Gamma\mathcal J\le
 \left(\rho+\frac{\pi^2}{4}\tau\right)
       \left(P-\sum_{i\in K}w_i\right)
 +\sum_{i\in K}\mathcal Q_u(\widehat f(i))w_i
 \le bP.
\]
The last inequality uses nonnegative weights and their common cost
bound $b$ from \eqref{st:spectral-cost}. This proves the first assertion.

For the second, write $\Delta=\Delta_\rho(f)$ and $\mu=\E f$.
Equation~\eqref{eq:entropy-remainder} gives
\[
 \Lambda\Var(g)-\Gamma\mathcal J
 =N-\Gamma\Delta
 +(\Gamma\log2-\Lambda)(\rho^2-\Var(g))-\Gamma\Psi(\mu).
\]
Since $\rho^2-\Var(g)\ge\rho^2\mu^2$ and
$\Psi(\mu)\le(\log2-1/2)\mu^2$, the transfer condition makes the last two terms nonnegative.
Thus the variance lost through a nonzero mean pays for its entropy
remainder. Divide by $b$ and use $N/b=u+m$ to conclude.
\end{proof}

\begin{lemma}\label{st:spectral-margins}
Write $\rho=\tanh u$.
In the compact range $\atanh(49/50)\le u\le8$, use the
parameter rule of Appendix~\ref{sec:compact-spectral} and take
$\bar\alpha=1-2e^{-u}$ for $u\le4$ and
$\bar\alpha=31/32$ for $u>4$ in \eqref{st:spectral-cost}.
Then $m>1/200000$ and $\Gamma/b<3000$.

For $u\ge8$, take $\bar\alpha=31/32$ and
$z=4$, $t=u-2-\tfrac12\log u$, and $\gamma=e^{2u-7}/u$.
Then $m>u/200$ and $\Gamma/b<4e^{2u-7}/u$.
In both ranges the parameters are admissible and satisfy the
transfer condition in Proposition~\ref{st:robust-spectral}.
\end{lemma}
The proof in Appendix~\ref{app:stability-inputs} translates the
normalized scalar margins into the production margin $m$
and the loss coefficient $\Gamma/b$ used here.

\subsection{Conditional variance around a marked coordinate}

The marked coordinate separates $f$ into two sections. Their means
need not be opposite when $f$ is unbalanced, so we retain both means
and use conditional variance to measure dependence on the remaining
coordinates.

For the marked-coordinate arguments, index the coordinates by
$\{0,\ldots,n-1\}$. For an increasing nonconstant source, choose
coordinate $0$ to have a largest singleton coefficient. Put
$a_0=\widehat f(0)=1-2\delta$, where $\delta=\delta(f)$.
Its sections $f(1,\cdot)$ and $f(-1,\cdot)$ have means
$\mu+a_0$ and $\mu-a_0$, respectively, so $|\mu|\le2\delta$.
For $i\ne0$, write $a_i=\widehat f(i)$ and
$b_i=\widehat f(\{0,i\})$. The sections are increasing, and
comparison with the constant functions $1$ and $-1$ gives
$0\le a_i\pm b_i\le2\delta\mp\mu$.
Consequently $|b_i|\le a_i\le2\delta$.

For $J\subseteq\{1,\ldots,n-1\}$, let
$r_\pm=\E[f(X)\mid X_0=\pm1,X_J]$ and put
\begin{equation}\label{eq:error-projection}
 U=\tfrac12[\Var(r_+)+\Var(r_-)],\qquad
 M_J=\sum_{i\in J}a_i.
\end{equation}
The means and singleton coefficients of these projections remain
$\mu\pm a_0$ and $a_i\pm b_i$. Apply
Lemma~\ref{as:strengthened-split} to $r_+$ and $r_-$ with
$K=\varnothing$ and parameter $4$, then average. Since
$0\le\partial_i r_\pm\le1$ and $|b_i|\le a_i$, this gives
\begin{equation}\label{eq:conditional-variance-bound}
 3U\le\tfrac12[\operatorname{Dir}(r_+)+\operatorname{Dir}(r_-)]
       +\frac{e^3}{2}\sum_{i\in J}(a_i^2+b_i^2)
 \le M_J+e^3\sum_{i\in J}a_i^2.
\end{equation}

Put $K=\{0\}\cup J$, $Q=\{0,\ldots,n-1\}\setminus K$,
$g=T_\rho f$, and $h=P_Kg$.
Conditional variance gives $\Var(P_Kf)=a_0^2+U$. After noise,
$E_J:=\Var(h)-\rho^2a_0^2\in[0,\rho^2U]$ and
$\operatorname{Dir}_J(h)\ge E_J$, since every mode counted by
$E_J$ contains a coordinate of $J$.
For a real parameter $\lambda$, apply
Lemma~\ref{as:strengthened-split} with parameter $\lambda-2$
in each section $X_0=\pm1$, revealing $J$ before $Q$.
The sectional mean derivatives are $\rho a_i\pm\rho^2b_i$;
averaging the two inequalities therefore yields
\begin{equation}\label{rt:section-fs}
 \operatorname{Dir}_{J\cup Q}(g)
 \ge\operatorname{Dir}_J(h)+\frac\lambda2\|g-h\|_2^2
 -\frac{e^{\lambda-3}}2\sum_{i\in Q}
       (\rho^2a_i^2+\rho^4b_i^2).
\end{equation}
Here projection and residual energies have disjoint Fourier supports.

\subsection{A marked-coordinate production estimate}

The marked coordinate contributes its entropy profile, while the
remaining variance supplies additional production. We combine these
two contributions using the conditional energy estimate
\eqref{rt:section-fs}. The retained-coordinate costs will again be
paid from the same production budget.

Use the coefficients $a_i,b_i$ above and the heights
$\ell_i=\ell_u(a_i)$ from \eqref{eq:profiles}.
Choose $0<z<u$ and $\lambda\ge2$, and put
\begin{equation}\label{eq:marked-parameters}
 \begin{gathered}
 \beta=F(u)/F(z),\quad J=\{i\ne0:a_i\ge\beta\},\quad
 c_J=\sum_{i\in J}a_i\ell_i,\\
 C=\frac{\rho^2(1+\rho^2)e^{\lambda-3}\beta}{2z},\qquad
 \omega=C+\rho.
 \end{gathered}
\end{equation}
It suffices to check the retained-coordinate cost at two influences:
\begin{equation}\label{ub:new-heads}
 \omega\ell_u(a)\ge\rho^2\left(1+\frac{\lambda-2}{6}
                      +\frac{(\lambda-2)e^3a}{6}\right)
 \qquad(a\in\{\beta,2\delta\}).
\end{equation}

\begin{proposition}[Marked-coordinate comparison]\label{st:robust-marked}
Let $f$ be increasing and nonconstant, let $u>0$, and suppose
$\vartheta\ge1$ satisfies $E_\rho(f)\le\vartheta H(\rho)$.
If \eqref{ub:new-heads} holds, or if $2\delta<\beta$, then
\[
 \vartheta\frac{D(T_\rho f)}\rho
 \ge a_0\ell_u(a_0)
 +\frac\lambda{2\omega}
     [\Var(T_\rho f)-\rho^2a_0^2].
\]
\end{proposition}
\begin{proof}
Put $P=\vartheta D(g)/\rho$. Lemma~\ref{st:dilation} gives
$\vartheta\mathcal D_i(g)\ge\rho a_i\ell_i$.
For an increasing Boolean source, $0\le\partial_i g\le\rho$ and
$\E\partial_i g=\rho a_i$, so
$\operatorname{Dir}_i(g)\le\rho^2a_i$.
Use the profile bound on coordinate $0$, the difference
$\vartheta\mathcal D_i(g)-\operatorname{Dir}_i(g)
\ge\rho a_i\ell_i-\rho^2a_i$ on $J$, and
$\vartheta\mathcal D_i(g)\ge\operatorname{Dir}_i(g)$ elsewhere.
The discarded profile budget is at most $P-a_0\ell_0-c_J$.
Since $|b_i|\le a_i$ and $\ell_u(a)/a$ decreases,
\[
 \sum_{i\in Q}(\rho^2a_i^2+\rho^4b_i^2)
 \le\rho^2(1+\rho^2)\frac\beta z(P-a_0\ell_0-c_J).
\]
Substitution in \eqref{rt:section-fs}, using $E_J\le\rho^2U$, gives
\begin{equation}\label{rt:marked-ledger}
 \omega(P-a_0\ell_0)\ge
 \frac\lambda2[\Var(g)-\rho^2a_0^2]
 +\omega c_J-\rho^2M_J
 -\frac{\lambda-2}{2}\rho^2U.
\end{equation}
By \eqref{eq:conditional-variance-bound}, the last three terms
are at least
\[
 \sum_{i\in J}a_i\left\{\omega\ell_u(a_i)
 -\rho^2\left(1+\frac{\lambda-2}{6}
                +\frac{(\lambda-2)e^3a_i}{6}\right)\right\}.
\]
The expression in braces is concave in $a_i$. As
$\beta\le a_i\le2\delta$, its endpoint bounds
\eqref{ub:new-heads} make the sum nonnegative.
If $2\delta<\beta$, the sum is empty. Division by $\omega$ proves
the claim.
\end{proof}

\subsection{Choosing the marked-coordinate parameters}

The parameters balance the discarded-coordinate cost against the
production surplus. Taking $z$ proportional to $u$ and keeping
$\lambda-2(u-z)$ bounded cancels the exponential factors in $C$,
leaving $C=O(1/u)$. At leading order, paying the smallest retained coordinate
requires $z>\lambda/6$, while a positive surplus requires
$(2-3\delta)\lambda>2u$ at the chosen distance cutoff.
These conditions leave a range of parameter choices. The following
choice puts the marked estimate directly in the form needed by the
production criterion.

\begin{corollary}\label{ub:join}
Take $z=19u/40-1/2$ and $\lambda=21u/20+2$.
For every increasing Boolean source, write $Z=Z_f(u)$ and
$\delta=\delta(f)$. Then
\begin{equation}\label{eq:marked-production-criterion}
 (1+Z_+)\frac{D(T_\rho f)}\rho
 \ge u\left(1+\frac\delta{40}-Z_+\right)
\end{equation}
in each of the following ranges:
\[
 \begin{array}{ll}
 \log96\le u\le8,& d_-(u)\le\delta\le1/64,\\
 u\ge8,& e^{-5u/4}\le\delta\le1/64.
 \end{array}
\]
\end{corollary}
The proof in Appendix~\ref{app:marked-parameters} checks the retained
costs and the variance surplus separately. The parameters remain
free in Proposition~\ref{st:robust-marked}; only this numerical
consequence uses their displayed values.

\subsection{Proof of the main theorem}\label{sec:global-stability}

We now verify the hypothesis of
Proposition~\ref{prop:production-stability}. The local comparisons
give a lower bound directly. Every remaining source has a production
surplus, supplied by the spectral or marked-coordinate estimate.

\begin{proposition}\label{st:high}\label{st:endpoint-order}
For every increasing Boolean function $f: \{-1,1\}^n\to \{-1,1\}$ and $49/50\le\rho<1$,
\[
 \Delta_\rho(f)\ge10^{-8}H(\rho)\delta(f).
\]
\end{proposition}
\begin{proof}
For constant functions, use $\Delta_\rho(f)=\Phi(\rho)\ge\rho^2/2$
and $\delta(f)=1/2$. The dictator case is exact.
For any other increasing source, put $u_0=\atanh(49/50)$,
$\delta=\delta(f)$, $u=\atanh\rho$, and
$Z=Z_f(u)$ as in \eqref{eq:normalized-deficit}.
We show that $Z<\delta/400000$ forces
\begin{equation}\label{eq:uniform-production}
 (1+Z_+)\frac{D(T_\rho f)}\rho
 \ge u\left(1+\frac\delta{800000}-90Z_+\right).
\end{equation}
The local comparisons cover the smallest distances, the spectral
comparison covers the largest, and the marked-coordinate estimate
fills the interval between them:
\begin{center}
\small
\renewcommand{\arraystretch}{1.15}
\begin{tabular}{@{}llll@{}}
\toprule
Range & Local & Marked coordinate & Spectral\\
\midrule
$u_0\le u\le4$ & $\delta\le e^{-u}$ & --- & $\delta\ge e^{-u}$\\
$4<u\le\log96$ & $\delta\le1/64$ & --- & $\delta\ge1/64$\\
$\log96<u\le8$ & $\delta\le d_-(u)$ & $d_-(u)\le\delta\le1/64$ & $\delta\ge1/64$\\
$u>8$ & $\delta\le e^{-5u/4}$ & $e^{-5u/4}\le\delta\le1/64$ & $\delta\ge1/64$\\
\bottomrule
\end{tabular}
\end{center}

In the four local regions, the lower bounds for $Z/\delta$ are
$1/600$, $1/400000$, $1/720$, and $1/4$, respectively.
Thus none occurs when $Z<\delta/400000$. The first and third bounds follow
from Lemma~\ref{rt:local-junction}, concavity of $M_\rho$, and
\eqref{eq:entropy-endpoint}. For the second, propagate the threshold
at distance $1/64$ downward: its coefficient is at least
$\rho^2/(40\cdot9217)>1/400000$, and $H(\rho)<1$.
The fourth is Lemma~\ref{app:endpoint-margins}.

In each spectral region, the distance cutoff gives
$\max_i\widehat f(i)=1-2\delta\le\bar\alpha$.
Using the corresponding parameters,
Lemma~\ref{st:spectral-margins} gives
\[
 \frac mu\ge\frac\delta{800000},\qquad
 \frac{\Gamma H(\rho)}{bu}<90.
\]
For $u\le8$, use $m>1/200000$, $\delta\le1/2$,
$\Gamma/b<3000$, $u>2$, and $H(\rho)<3/50$.
For $u>8$, the growing margin $m>u/200$ suffices, and
\eqref{eq:entropy-endpoint} bounds the normalized loss by
$4e^{-7}(2+1/u)/u<1$.
With $\vartheta=1+Z_+$, Proposition~\ref{st:robust-spectral}
therefore gives \eqref{eq:uniform-production}, using
$-\Delta_\rho(f)\ge-H(\rho)Z_+$.

In the marked-coordinate regions, Corollary~\ref{ub:join} gives
the stronger bound \eqref{eq:marked-production-criterion}.
The table covers every source, so
Proposition~\ref{prop:production-stability} applies with
$a=1/800000$, $B=90$, and $\eta=1/400000$.
Since $10^{-8}<\min\{\eta,1,a/(B+3)\}$, it proves the claim.
\end{proof}

\begin{proof}[Proof of Theorem~\ref{st:natural-intro}]
For increasing $f$, Corollary~\ref{cor:low-stability} and
Propositions~\ref{st:middle} and~\ref{st:high} give
$\Delta_\rho(f)\ge10^{-8}\rho^2H(\rho)\delta(f)$.
Lemma~\ref{st:compression} therefore gives, for every Boolean $f$,
\begin{equation}\label{eq:global-compressed-gap}
 \Delta_\rho(f)\ge
 \frac{\rho^2(1-\rho^2)H(\rho)}{10^8(2-\rho^2)}\,\delta(f).
\end{equation}
For $\rho\ge457/500$, we have
$(1-\rho^2)/(2-\rho^2)\ge1-\rho$, so
\eqref{eq:global-compressed-gap} gives
$\Delta_\rho(f)\ge10^{-8}\rho^2(1-\rho)H(\rho)\delta(f)$.
For $0<\rho\le457/500$, the same lower bound follows directly
from Corollary~\ref{cor:low-stability}, since
$(1-\rho)H(\rho)<1$.
Since the hypothesis gives $\Delta_\rho(f)\le\varepsilon$,
these lower bounds imply the stated probability estimates.
For increasing $f$, the singleton coefficients are nonnegative,
so a nearest signed dictator can be chosen with positive sign.

To examine sharpness at the endpoints, take $f$ to be
three-variable majority, which is increasing and balanced,
with $\delta(f)=1/4$. Put
$r=(3\rho-\rho^3)/2$ and $s=(\rho+\rho^3)/2$.
The posterior magnitude $|T_\rho f|$ is $r$ with probability
$1/4$ and $s$ with probability $3/4$, hence
\[
 \frac{\Delta_\rho(f)}{\delta(f)}
 =H(r)+3H(s)-4H(\rho).
\]
As $\rho\downarrow0$, the expansion
$H(t)=\log2-t^2/2+O(t^4)$ gives
\[
 \frac{\Delta_\rho(f)}{\delta(f)}
 =\frac{\rho^2}{2}+O(\rho^4).
\]
Both denominators in Theorem~\ref{st:natural-intro} are
asymptotic to $(\log2)\rho^2$ at this endpoint.
Thus the order of both bounds as $\rho\downarrow0$ is optimal,
already among increasing functions.

As $\rho\uparrow1$, we have
$1-r=O((1-\rho)^2)$ and $1-s\sim2(1-\rho)$.
Consequently, $H(r)=o(H(\rho))$ and $H(s)\sim2H(\rho)$, so
\[
 \frac{\Delta_\rho(f)}{\delta(f)}
 \sim2H(\rho)\qquad \text{as } \rho\uparrow 1.
\]
Since $\rho^2\to1$, the denominator $\rho^2H(\rho)$ in the
increasing-function bound also has optimal order at this
endpoint.
\end{proof}

The extra factor $1-\rho$ in the general bound arises from
quantitative compression. We do not determine whether this
factor can be removed.

\begin{samepage}
\section*{Use of AI tools}

The authors formulated the approach and used ChatGPT to develop the proofs, prepare the verification code, and revise the exposition. The authors take full responsibility for the mathematical claims, proofs, computations, and references in this paper.
\end{samepage}

\appendix

\section{Quantitative inputs for stability}\label{app:stability-inputs}

This appendix verifies the spectral and marked-coordinate parameter
choices used in Section~\ref{sec:radius-tail}, followed by the local
bound that completes the coverage near the noiseless endpoint.

\subsection{Scalar cost bounds}\label{app:scalar-cost}

\begin{lemma}[Convex cost and explicit inverse
{\cite[Lemmas~7.7 and~7.8]{VuTranCK}}]
\label{as:convex-majorant}\label{as:explicit-inverse}
Choose either expression in the minimum defining $\psi$ in
\eqref{eq:retained-variance}, and substitute it into
$\mathcal Q_u(F(u)/F(v))$ from \eqref{eq:spectral-profile}. Replacing
$[v-R(v)]/v$ by
$1-R'(c)+[cR'(c)-R(c)]/v$, for $c>0$, gives a convex upper
bound on $0<v\le u$. Alternatively, replacing the entire arcsine
term $\rho[v-R(v)]/v$ by $\pi^2/(4v)$ also gives a convex upper bound.
If $u>2$, $e^{2-u}\le a\le1$, and
$v=u+(2u+1)\log a/(4u+1)$, then
\[
 F^{-1}(F(u)/a)\le v\le u,\qquad
 \frac{F(u)}{F(v)}\le e^{-2(u-v)}\frac{2u+1}{2v+1}.
\]
\end{lemma}
These estimates reduce each closed height interval to its endpoints
and avoid numerical inversion of the profile.

\subsection{The analytic spectral margin}

The two expressions defining $\psi$ in \eqref{eq:retained-variance}
give a quadratic and a logarithmic upper bound for the retained cost.
For the analytic estimate below, use $\pi^2/(4v)$ for the arcsine
contribution, the logarithmic bound for $a\le1/8$, and the quadratic
bound for $a\ge1/8$. At the common boundary the logarithmic cost is
smaller, since $\log2>31/45$.
Each branch is convex, so its maximum is attained at an endpoint.
Only the two outer heights and $\ell_u(1/8)$ need to be checked.
The fixed influence cutoff leaves a uniform gap between these costs
and $N/u$, producing a production margin proportional to $u$.

\begin{proof}[Proof of the analytic part of Lemma~\ref{st:spectral-margins}]
Use the stated parameters in \eqref{eq:spectral-parameters} and
\eqref{st:spectral-cost}.
Then $\gamma=e^{2t-3}$, so admissibility holds.
We may use the fixed cutoff $a\le31/32$.

The elementary bounds $e^3>20$, $e^2>7$, $8/3<e<14/5$,
and $9/13<\log2<7/10$ will suffice.
For $u\ge8$, formula \eqref{rt:scaled-profile} gives
$2u+1\le e^{2u}F(u)\le(1001/1000)(2u+1)\le15u/7$,
while $e^8F(4)\ge9$.
Moreover,
\[
 t<u-3,\quad t/u\ge99/160,\quad \gamma>1000,\quad
 k(t)\rho^2>t-1/32,\quad \gamma C_*\Psi(\rho)<1/250.
\]
For the spectral error, use
$t^2/\gamma<u(u-3)^2e^{7-2u}<1/40$ and
$1-\rho^2\le4e^{-2u}$.
For the entropy charge, use $C_*<12/7$, $\Psi(\rho)<H(\rho)$,
and $e^7>1000$. The expressions involving an exponential decrease
from $u=8$; $(2+\tfrac12\log u)/u$ also decreases.
In particular, $k(t)\rho^2/u>3/5$, and the baseline
$N/u$ from \eqref{st:spectral-cost} exceeds $29/20$.
The transfer condition follows from
$u/\gamma<1/125$, $1-\rho^2<1/100000$, and $C_*>1$.

For discarded coordinates,
$\tau<\gamma F(u)/(4F(4))\le5e/84<1/6$.
Hence $\rho+\pi^2\tau/4<17/12<29/20$.
Lemma~\ref{as:convex-majorant} leaves three retained heights.

At $v=4$, one has $a\le(u/4)e^{-2(u-4)}<1/100$ and
$\log(4/a)>2t-3/2$. Since $k(2)>3/2$,
$\Delta_2\,2(1-a/4)/\log(4/a)<1$.
Also $\Delta_1-\Delta_2<2$, so the logarithmic majorant is
less than $101/80<4/3$.

For the quadratic branch, combine the degree corrections before
estimating their error. Their numerator satisfies
\begin{align*}
 1+a\Delta_1+\tfrac12(1-a)\Delta_2
 & = k(t)\left[a\rho^2+\tfrac12(1-a)\rho^4\right]+a[1-\rho^2k(1)]+(1-a)[1-\rho^4k(2)/2]\\
 &\le \tfrac12 k(t)(1+a)+2/\gamma+8e^{-2u}.
 \end{align*}
Indeed, the two terms in the convex combination are at most
$1/\gamma+4e^{-2u}$ and $2/\gamma+8e^{-2u}$.
The common error is below $1/400$. Thus, whenever $v\ge u-6/5$,
the quadratic majorant is bounded by
\[
 \frac{\pi^2k(t)(1+a)}{8v}+\frac1{1000}.
\]
At $a=1/8$, the bound $\beta<1/100$ gives $v>4$.
The profile derivative in Lemma~\ref{as:profile-convex} then gives
$v>u-(4/7)\log8>u-6/5>t$, so the majorant is below
$45/32+1/1000<17/12$.

At the top endpoint, put $a=31/32$ and $v=\ell_u(a)$.
The profile derivative gives
$u-v\le(4/7)\log(32/31)<1/50$, and hence
\[
 \frac{u(1+a)}{2v}\le\frac{75}{76},\qquad
 \rho^2-\frac{u(1+a)}{2v}>\frac1{80}.
\]
For the second inequality, use $e^8>2500$.
Since $k(t)/u>3/5$, $\pi^2/4>12/5$, and $v>u-6/5$,
the normalized baseline $N/u$ exceeds the top majorant by more than
\[
 \frac9{500}-\frac1{2000}-\frac1{1000}>\frac1{80}.
\]
The two subtracted terms bound the entropy charge and quadratic error.
The other costs have the larger margin $29/20-17/12=1/30$.
Thus $N/u-b>1/80$ and $b<N/u<5/2$, giving $N/b-u>u/200$.
Finally $C_*<12/7$ and $b\ge\rho>97/100$ give
$\Gamma/b<4e^{2u-7}/u$.
\end{proof}

\subsection{Normalization of the compact spectral margins}
\begin{proof}[Proof of the compact part of Lemma~\ref{st:spectral-margins}]
The compact certificate described in Appendix~\ref{sec:compact-certificate}
covers 24 intervals and gives normalized margins
greater than $1/1000$ for the discarded-coordinate inequality and
$1/100$ for each of the four retained-height inequalities.
We record their positive normalization factors to translate these
signs into a production margin.

Put $w=1/\gamma$, $\xi=e^{-9u/20-1/4}$, and
$R_0=\frac12e^{53u/100-9/20}$.
At height $v$, write $a=F(u)/F(v)$.
Let $\widetilde{\mathcal Q}_u(v)$ denote the chosen
supporting-line majorant of $\mathcal Q_u(F(u)/F(v))$
from Appendix~\ref{sec:compact-spectral}, using the explicit
upper bound for $a$ at the top height.
The retained normalization factor is
\[
 v(1+tw)(1+w)(1+2w)
 \begin{cases}\log(4/a),&\text{logarithmic cost},\\
 1,&\text{quadratic cost}.
 \end{cases}
\]
The discarded-coordinate factor is
$(1+R_0+t\xi)(1+t\xi)zF(z)$.
Multiplication by these factors gives exactly the expressions
checked by \path{compact_uniform_verify.py}, applied to
$N-u\widetilde{\mathcal Q}_u(v)$ and
$N-u(\rho+\pi^2\tau/4)$, respectively.

The compact parameter estimates give $w<1/7$, $tw<2/5$,
$t\xi<1$, and $R_0<55/2$, using $e^4<55$.
Also $zF(z)<\log2<7/10$. Thus the discarded factor is below $56$.
Since $z>1/3$, $F(z)<3$, and
$F(u)\ge(2u+1)e^{-2u}$, each logarithmic height has
$a\ge e^{-2u}$ and hence $\log(4/a)<18$.
The retained factor is below
$8(7/5)(8/7)(9/7)\cdot18<306$.
It follows that $N-ub>1/56000$.
Furthermore $b<N/u<(\pi^2/4)t/u<3$, so $m>1/168000>1/200000$.
Finally $e^8<3000$ and $e^4<55$ give $\gamma<1600$ and hence
$\Gamma/b<(1600)(12/7)/(97/100)<3000$.
All additional endpoint constants and stronger recorded margins are
checked by \path{check_stability_constants.py}.
\end{proof}

\subsection{The marked-coordinate parameter bounds}\label{app:marked-parameters}

The following elementary bounds control the discarded cost and the
entropy loss on the whole interval; no subdivision certificate is needed.

\begin{lemma}\label{app:neighborhood-constants}
With the parameters of Corollary~\ref{ub:join}, every $u\ge9/2$
satisfies $Cz>7/25$, $C(u+4/7)<1$, and
\[
 \lambda\bigl[8e^{-2u}+(u-3/4)e^{-u}\bigr]<\frac{29}{100}.
\]
\end{lemma}
\begin{proof}
The exponential series gives $90<e^{9/2}<96$,
$22/3<e^2<15/2$, and $20<e^3<21$.
Hence \eqref{rt:scaled-profile} gives
$e^{2u}F(u)<(1001/1000)(2u+1)$.
Also $e^{2z}>4u+1$: it holds at $u=9/2$ by $e^3>20$,
and $e^{2z}/(4u+1)$ increases. Since $2z+1=19u/20$,
the same profile bound gives
\[
 \frac{e^{2u}F(u)}{e^{2z}F(z)}
 \ge\frac{2u+1}{2z+1}\tanh z>\frac{40}{19}.
\]
The cancellation $\lambda-3-2u+2z=-2$ yields
$Cz=\rho^2(1+\rho^2)e^{-2}e^{2u}F(u)/(2e^{2z}F(z))$.
Use $\rho^2(1+\rho^2)/2>1-8/90^2$ to obtain
$Cz>(1-8/90^2)(2/15)(40/19)>7/25$.
For the upper bound,
\[
 C(u+4/7)<\frac3{22}\frac{1001}{1000}
       \frac{2u+1}{2z+1}\frac{u+4/7}{z}.
\]
Both ratios decrease. At $u=9/2$, the right side is
$7384/7467<1$.

Both summands in the final inequality decrease for $u\ge9/2$.
For the second, the logarithmic derivative before multiplication
by $e^{-u}$ is less than $1/u+1/(u-3/4)<1$.
At $u=9/2$, use $e^{9/2}>90$; the resulting upper bound is
$185879/648000<29/100$.

All constants here and in Corollary~\ref{ub:join} follow from
rational arithmetic and
\[
 \sum_{j=0}^{32}\frac{x^j}{j!}<e^x<
 \sum_{j=0}^{32}\frac{x^j}{j!}
       +\frac{x^{33}}{33!}\frac1{1-x/34}\qquad(0<x<34).
\]
The upper bound sums a geometric majorant of the remainder.
The logarithm bound follows from the first three terms of
$\log2=2\sum_{j\ge0}[(2j+1)3^{2j+1}]^{-1}$.
\end{proof}

\begin{proof}[Proof of Corollary~\ref{ub:join}]
Fix $u\ge9/2$ and $0<\delta\le1/64$, and use the stated parameters.
We have $z>1>\rho$ and $\lambda>2$.
The scaled profile bound gives $\beta<(5/2)e^{-21u/20-1}$,
so in particular $\beta<1/32$ and $\beta<d_-(u)/2$.
Lemma~\ref{app:neighborhood-constants} gives
$Cz>7/25$ and $C(u+4/7)<1$ throughout $u\ge9/2$.
Thus $\omega=C+\rho<7/5$, while $\omega>\rho>97/100$
and $\log2>2/3$ give $2\omega\log2>1$.

To check \eqref{ub:new-heads}, observe that its left side minus its
right side is concave in $a$. Since $2\delta\le1/32$, it suffices
to check $a=\beta$ and $a=1/32$; if $2\delta<\beta$, there is no
retained coordinate to check. At the first endpoint the height is $z$; at the second,
$\ell_u(1/32)>u-2$. Indeed, \eqref{rt:scaled-profile} and $e^2>22/3$ give
$F(u-2)/(32F(u))>1$, with the lower bound increasing from $9/2$.

When $\ell_u(a)\ge v$, the left side minus the right side of
\eqref{ub:new-heads} is bounded below by
\[
 v-1-\frac{7u}{40}+\frac{7v}{25z}
       -2ve^{-2u}-\frac{7u}{40}e^3a.
\]
Take $(v,a)=(z,\beta)$ at the first endpoint. Its linear part
is $3u/10-61/50$. The two losses are bounded by
$2ze^{-2u}$ and
$(7/19)(1001/1000)(u+1/2)e^{2-21u/20}$, which decrease.
At $u=9/2$, use $e^{9/2}>90$ and $e^{109/40}>15$;
the resulting lower bound exceeds $1/200$.
For the fixed endpoint, take $(v,a)=(u-2,1/32)$ and use
$e^3<21$. The linear part and $v/z$ increase, while the exponential
loss decreases; evaluation at $9/2$ gives a lower bound $3/5$.

The marked profile satisfies
\begin{equation}\label{rt:profile-shift}
 u-(1-2\delta)\ell_u(1-2\delta)\le2\delta(u+4/7).
\end{equation}
Indeed, for $3/4\le a\le1$, Lemma~\ref{as:profile-convex} gives
$\ell_u(a)\ge u+\log a>4$, and its logarithmic derivative bound gives
\[
 \frac{d}{da}[a\ell_u(a)]
 =\ell_u(a)-\frac{F(\ell_u(a))}{F'(\ell_u(a))}\le u+4/7.
\]
Integrate from $1-2\delta$ to $1$.

Write $\mu=\E f$, $\Delta=\Delta_\rho(f)$,
$Z=\Delta/H(\rho)$, and $P=(1+Z_+)D(T_\rho f)/\rho$.
Dropping $\mathcal J\ge0$ in \eqref{eq:entropy-remainder}
and using $|\mu|\le2\delta$ gives
\[
 \Var(g)\ge\frac{\Phi(\rho)-\Psi(\mu)-\Delta}{\log2},
 \qquad \Psi(\mu)\le2\delta^2\log2.
\]
Substitute these bounds and \eqref{rt:profile-shift} into
Proposition~\ref{st:robust-marked}. This gives
\[
 \begin{aligned}
  \omega(P-u)+\frac{\lambda\Delta}{2\log2}
  \ge{}&\delta\bigl[
  2\lambda\rho^2-(2\rho^2+1)\lambda\delta
  -2\omega(u+4/7)\bigr]\\
  &-\frac{\lambda\Psi(\rho)}{2\log2}.
 \end{aligned}
\]
Use $\rho^2\ge1-4e^{-2u}$, $C(u+4/7)<1$, and $\delta\le1/64$.
Since $(2-3/64)\lambda-2u-22/7=13u/256+171/224$, for every
$0<d\le\delta$ this gives
\begin{equation}\label{eq:marked-common-margin}
 \omega(P-u)+\frac{\lambda\Delta}{2\log2}
 \ge\delta\left[
 \frac{13u}{256}+\frac{171}{224}-8\lambda e^{-2u}
 -\frac{\lambda\Psi(\rho)}{2d\log2}\right].
\end{equation}
First suppose $\log96\le u\le8$ and $\delta\ge d_-(u)$.
Put $q=e^{-2u}$. The entropy and noise bounds give
\[
 0\le\frac{\Psi(\rho)}{\log2}
 \le(3u-5/2+8q)q<3(u-3/4)q.
\]
Take $d=d_-(u)=(3/2)e^{-u}$ in \eqref{eq:marked-common-margin}.
Lemma~\ref{app:neighborhood-constants} bounds its two losses by $29/100$.
The remaining bracket is at least $13u/256+171/224-29/100$,
which increases from $62891/89600>7/10$ at $u=9/2$.
As $\omega<7/5$ and $\lambda/(2\omega\log2)<\lambda$, we conclude
\[
 P\ge u+\delta/2-\lambda\Delta_+
   \ge u+u\delta/40-\lambda\Delta_+.
\]
The second inequality uses $u\le8$.

For $u\ge8$ and $\delta\ge e^{-5u/4}$, instead take $d=e^{-5u/4}$.
Use $\Psi(\rho)\le H(\rho)\le(2u+1)e^{-2u}$ and $\log2>2/3$.
The two losses in \eqref{eq:marked-common-margin} total at most
\[
 8\lambda e^{-2u}+\frac{3\lambda}{4}(2u+1)e^{-3u/4}<2/5.
\]
Both terms decrease for $u\ge8$; the bound at $8$ follows from
$e^6>400$ and $e^8>2500$.
Since $13u/256+171/224-2/5>u/20$, the same bounds on $\omega$
and the loss coefficient give
$P\ge u+u\delta/40-\lambda\Delta_+$.

Finally, $\lambda<3u$ and $H(\rho)<1/3$ throughout $u\ge\log96$,
so $\lambda\Delta_+\le uZ_+$. The two ranges therefore satisfy
$P\ge u(1+\delta/40-Z_+)$, as required.
\end{proof}

\subsection{The local bound at the endpoint}

\begin{lemma}\label{app:endpoint-margins}
Put $u=\atanh\rho\ge8$ and $d(u)=e^{-5u/4}$.
For every Boolean source with $\delta(f)\le d(u)$,
$\Delta_\rho(f)\ge u(1-\rho)\delta(f)/2$.
\end{lemma}
\begin{proof}
Put $\delta=\delta(f)$. The dictator case is immediate.
For $0<\delta\le d(u)$, compare the entropy scales at
$d=e^{-5u/4}$ and $p=(1-\rho)/2$. Since
$d\le p+\rho d\le1/2$, we have
$H_{\rm b}(p+\rho d)\ge H_{\rm b}(d)$. Using
$s[\log(1/s)+1-s]\le H_{\rm b}(s)\le s[\log(1/s)+1]$
and $\log(1/p)\le2u+2p$ gives
\[
 \begin{split}
 \frac{M_\rho(1-2d)}{2dp}
 &\ge2(1-p)(5u/4+1-d)
       -(1+2p/d)(2u+1+2p)\\
 &\ge\frac u2+1-8u e^{-3u/4}>\frac u2.
 \end{split}
\]
Indeed, the leading terms give $u/2+1$. Since
$d,p,p/d\le e^{-3u/4}$ and $p<1/2$, their correction is at most
$(13u/2+10)e^{-3u/4}\le8u e^{-3u/4}$.
This decreases for $u\ge8$ and is less than $64/400<1$ at $8$.
Equation~\eqref{eq:local-distance-conversion} therefore gives
$\Delta_\rho(f)\ge u(1-\rho)\delta/2$ for $\delta\le d(u)$,
without a monotonicity assumption.
\end{proof}

\section{Certificates and reproducibility}
\label{sec:certificates}\label{app:certificate}
The checks certify strict scalar inequalities and coverage of their
domains. Low-correlation signs use rational series bounds; other
transcendental signs use outward-rounded Arb intervals enclosing the
exact values. A sign passes only when its whole enclosure has that
sign. Exact rational and algebraic arithmetic checks coverage, and
subdivisions retain both children unless infeasibility is proved.
Convexity justifies the supporting lines; their integrals require
positive denominators and the stated final signs. The compact inverse
profile is bounded analytically, and separate scripts check the finite
constants in Section~\ref{sec:radius-tail}.

\subsection{Low-correlation margin checks}\label{app:low-margins}

For $\alpha\le69/100$, $|\mu|\le61/100$, and
$3/5\le\rho\le457/500$, the quadratic Fourier calculation in
\cite[proofs of Proposition~3.15 and Lemma~A.2]{VuTranCK} gives
$\Delta_\rho(f)\ge cV(\rho)+\mu^2/100$, where
\[
 c=\frac{57}{76-49\rho^2},\qquad
 V(\rho)=(1-\rho)
 \left(1+\rho-\frac{193027}{250000}\rho^2\right)
 -\frac{76-49\rho^2}{57}H(\rho).
\]
The companion's proof shows that $V$ decreases on this interval.
The endpoint check $V(457/500)>1/50000$ therefore gives
$V(\rho)>1/50000$ throughout. Since $c>3/4$, this proves the
deficit bound in \eqref{st:low-energy-input}. At $\rho=3/5$,
the same Fourier estimate gives $A_{3/5}(f)<9/50-1/1000$.
The local anchor satisfies
$M_{457/500}(69/100)>(457/500)^2/1000$. For the large-mean branch,
$H(61/100)\le347/500-(61/100)^2/2-(61/100)^4/12<1/2-1/1000$,
using $\log2<347/500$ and the entropy series.
The base comparisons and propagation are checked by
\path{low_correlation_verify.py}; the stronger margins used here
are enforced by \path{check_stability_constants.py}.

\subsection{Intermediate-range margin checks}\label{app:middle-margins}

The intermediate source cover uses the rational parameters in
Table~\ref{tab:middle-parameters}. Every finite decimal denotes an exact rational number.

\begin{table}[ht]
\centering\small
\begin{tabular}{@{}ccc@{}}
\toprule
Band & $M$ & $A_*$\\
\midrule
$[.914,.95]$ & $.82523844$ & $.66996$\\
$[.95,.97]$ & $.787706016$ & $.72245$\\
$[.97,.974]$ & $.77592593$ & $.73625$\\
$[.974,.975]$ & $.772609873$ & $.73996$\\
$[.975,.98]$ & $.753070335$ & $.76043$\\
\bottomrule
\end{tabular}
\vspace{3mm}
\caption{Parameters for the intermediate correlation bands.}
\label{tab:middle-parameters}
\end{table}

The scalar verifier uses
$\chi_\rho(t)=-2H(t)/(1-t)+2\kappa_\rho(1+t-\rho^3)$,
for $0\le t<1$. At both correlation endpoints,
its certified supporting line has values below $-1/4000$ at
$t=0,1$. By concavity, $\chi_\rho$ lies below this supporting
line, whose values are below $-1/4000$ throughout $[0,1]$.
Hence $\chi_\rho(t)\le-1/4000$,
which is exactly \eqref{eq:middle-entropy-slack} for $t<1$.
The local endpoint signs and clock margins in
\eqref{eq:middle-local-clock-slack} are recomputed from the same
source data by \path{check_stability_constants.py}.

\subsection{The compact spectral parameters}\label{sec:compact-spectral}

For $45951/20000\le u\le8$, use the explicit rule of
\cite[Section~7.6]{VuTranCK}:
\[
 z=\log u+u/25-2/5,\quad t=49u/100+7/5,\quad
 \gamma=\tfrac12e^{49u/50-1/5}+e^{9u/20+1/4}.
\]
That section proves admissibility and the transfer condition
analytically; its normalized transfer margin exceeds $26/875$.
For the stronger stability margins, take the influence cutoff
$c(u)=1-2e^{-u}$ through $u=4$ and $c(u)=31/32$ thereafter.
Here $c(u)$ is the cutoff $\bar\alpha$ in
\eqref{st:spectral-cost}.
Both values are checked at $u=4$. Set
\[
 v_+(u)=u+\frac{2u+1}{4u+1}\log c(u).
\]
Since $c(u)>e^{-1/4}>e^{2-u}$, Lemma~\ref{as:explicit-inverse},
applied at $a=c(u)$, gives $\ell_u(c(u))\le v_+(u)$.
As $\ell_u$ is increasing, every retained influence $a\le c(u)$
therefore satisfies $\ell_u(a)\le v_+(u)$.
Moreover, $c(u)>e^{-1/4}$ and $u>9/4$ give
$v_+(u)>u-1/4>2u/3$, while the explicit parameter rule gives
$0<z<2u/3$.

Use the logarithmic upper bound from Lemma~\ref{as:convex-majorant} on
$[z,2u/3]$, with the tangent to $R$ at $z$, and the quadratic
bound on $[2u/3,v_+(u)]$, with the tangent at $v_+(u)$.
Each bound is convex in height and needs only its two endpoints.
At the top height, use the explicit upper bound for
$F(u)/F(v_+(u))$; both influence costs increase with this argument.
These choices define the four retained-height checks used in the
normalization argument of Appendix~\ref{app:stability-inputs}.

\subsection{The compact certificate}\label{sec:compact-certificate}

The file \path{compact_uniform_certificate.json} covers
$[45951/20000,8]$ by 24 closed rational intervals in $u$: eight
intervals through $4$, followed by 16 intervals of length $1/4$.
Since $45951/20000<\atanh(49/50)$, these intervals, together
with the stated influence cutoffs, cover the full compact range
required in Lemma~\ref{st:spectral-margins}.

The verifier \path{compact_uniform_verify.py} checks the four height
inequalities per interval at 256 bits, using fourth-order Taylor
enclosures. The independent audit \path{compact_uniform_audit.py}
uses separately derived unscaled formulas and 384-bit arithmetic over
the full extended height domain. Exact adjacency checks verify the
partition, and \path{check_stability_constants.py} extracts the
stability margins.

\subsection{Files and replay scope}

The repository for the certificates, verification code, and input data is \url{https://github.com/vukhacky/CK-stability}.
The repository's \path{README.md} gives installation
instructions; \texttt{python3 verify.py} checks the manifest and
replays all twelve numerical stages. Assertions must remain enabled.
No file from the companion package is required.


\begin{thebibliography}{18}
\setlength{\itemsep}{0.3em}

\bibitem[Chen et~al.(2026)]{IndependentCK2026}
Z.~Chen, A.~Gohari, A.~Javanmard, H.~Lin, V.~Mirrokni,
C.~Nair, and D.~P.~Woodruff.
\newblock A proof of the most informative {Boolean} function conjecture.
\newblock arXiv:2609.24931, 2026.

\bibitem[Courtade and Kumar(2014)]{CourtadeKumar2014}
T.~A. Courtade and G.~R. Kumar.
\newblock Which {Boolean} functions maximize mutual information on noisy
inputs?
\newblock \emph{IEEE Trans. Inf. Theory}, 60\penalty0
(8):\penalty0 4515--4525, 2014.

\bibitem[Falik and Samorodnitsky(2007)]{FalikSamorodnitsky2007}
D.~Falik and A.~Samorodnitsky.
\newblock Edge-isoperimetric inequalities and influences.
\newblock \emph{Combin. Probab. Comput.}, 16\penalty0
(5):\penalty0 693--712, 2007.

\bibitem[O'Donnell(2014)]{ODonnell2014}
R.~O'Donnell.
\newblock \emph{Analysis of {Boolean} Functions}.
\newblock Cambridge Univ. Press, 2014.

\bibitem[Samorodnitsky(2016)]{Samorodnitsky2016}
A.~Samorodnitsky.
\newblock On the entropy of a noisy function.
\newblock \emph{IEEE Trans. Inf. Theory}, 62\penalty0
(10):\penalty0 5446--5464, 2016.

\bibitem[Yu(2023)]{Yu2023}
L.~Yu.
\newblock On the {$\Phi$}-stability and related conjectures.
\newblock \emph{Probab. Theory Relat. Fields}, 186\penalty0
(3--4):\penalty0 1045--1080, 2023.

\bibitem[Yu(2026)]{Yu2026Local}
L.~Yu.
\newblock Local optimality of dictator functions with applications to
{Courtade--Kumar} and {Li--M\'edard} conjectures.
\newblock arXiv:2410.10147v5, 2026.

\bibitem[Vu and Tran(2026)]{VuTranCK}
V.~K.~Ky and T.~Tran.
\newblock Dictators are most informative.
\newblock arXiv:2609.24184, 2026.


\end{thebibliography}
\end{document}